\documentclass[11pt,a4paper]{article}
\usepackage[T1]{fontenc}
\usepackage[utf8]{inputenc}
\usepackage{lmodern,microtype}
\usepackage{graphicx}
\usepackage[margin=27mm,headheight=14pt]{geometry}
\usepackage{amsmath,amssymb,amsthm,mathtools}
\usepackage{enumitem,booktabs,longtable,array}
\usepackage{xcolor}
\usepackage[colorlinks=true,linkcolor=blue!45!black,citecolor=blue!45!black,urlcolor=blue!45!black]{hyperref}
\hypersetup{pdftitle={A Temporal Splitting Theorem for Chronological Spaces},pdfauthor={Ali Bleybel},pdfsubject={Temporal foliations of chronological spaces}}
\usepackage{fancyhdr}
\setlist{topsep=4pt,itemsep=3pt}
\newtheorem{theorem}{Theorem}[section]
\newtheorem{lemma}[theorem]{Lemma}
\newtheorem{proposition}[theorem]{Proposition}

\theoremstyle{definition}
\newtheorem{definition}[theorem]{Definition}
\newtheorem{assumption}[theorem]{Assumption}
\theoremstyle{remark}
\newtheorem{remark}[theorem]{Remark}

\newcommand{\R}{\mathbb R}
\newcommand{\Q}{\mathbb Q}
\newcommand{\N}{\mathbb N}
\newcommand{\Z}{\mathbb Z}
\newcommand{\cl}[1]{\overline{#1}}
\newcommand{\bd}{\partial}
\newcommand{\Ip}{I^+}
\renewcommand{\Im}{I^-}
\newcommand{\Jp}{J^+}
\newcommand{\Jm}{J^-}
\newcommand{\leqC}{\leq_c}
\newcommand{\interior}{\operatorname{int}}

\title{A Temporal Splitting Theorem for Chronological Spaces}
\author{Ali Bleybel\thanks{E-mail: \href{mailto:bleybel@ul.edu.lb}{bleybel@ul.edu.lb}.}\\[3pt]
\small Faculty of Sciences (I), Lebanese University, Beirut, Lebanon}
\date{}
\begin{document}
\maketitle
\begin{abstract}
We obtain a generalization of Geroch's Theorem on temporal foliations in a globally hyperbolic space-time to chronological spaces under suitable topological and order-theoretic requirements. The foliation can be chosen to contain a given relatively compact antichain in one slice. We first prove the corresponding extension result for countable causal sets, and then give a direct proof for chronological spaces using past sets. Furthermore, we show that any chronological space-time admits a generalized Cauchy hypersurface containing a prescribed compact achronal set: it is crossed once by every inextendible timelike curve which is also a maximal chronological chain.
\end{abstract}
\section{Introduction}
In this paper we introduce results on chronological spaces, generalizing well known results for space-times.

The main result is a ``temporal splitting'' theorem for chronological spaces, provided they satisfy suitable topological and order-theoretic requirements. In addition, a given antichain contained in a compact set can be included in one of the slices. As an application of the obtained result, we obtain a proof of Geroch's Theorem~\cite{Geroch1970} on the existence of a splitting of a globally hyperbolic space-time $M$ into Cauchy slices.

Another result concerns a generalized notion of a ``Cauchy hypersurface;'' more precisely, we show that for a space-time that is chronological, there is a hypersurface $\Sigma$ crossed by any inextendible timelike curve which is also a maximal chronological chain. Moreover, $\Sigma$ can be chosen to contain a given compact achronal set.

In addition to their intrinsic interest, these developments provide a new perspective on the notion of a Cauchy hypersurface and its relation with the corresponding notion in the context of causal set theory (keeping in mind that the resemblance is far from perfect).

The recovery of continuous geometry and topology from discrete causal sets was considered in~\cite{BombelliMeyer1989,MajorRideoutSurya2007}. Here we first prove that a finite antichain in a countable causal set is contained in a quasi-Cauchy slice. In the continuous case we construct past sets with compact bounds in both time directions, insert further slices between their boundaries, and complete the resulting ordered family. The discrete and continuous constructions are proved separately. In particular, the continuous proof uses neither volume functions nor the foliation theorem of~\cite{BleybelZaiour2018}.

The causal-space framework was introduced in~\cite{KronheimerPenrose1967}. The relation between causal orders and domain theory is studied in~\cite{MartinPanangaden2006}; the extension of compact achronal sets in globally hyperbolic space-times is considered in~\cite{BernalSanchez2006}. Foliations by antichains were considered in~\cite{BleybelZaiour2018}. The weaker assertion that a poset admits an order-compatible, linearly indexed partition into antichains preserving a prescribed antichain follows from the classical order-extension theorem~\cite{Szpilrajn1930}. Continuous order-preserving functions on closed ordered spaces also have an established theory; see~\cite{Minguzzi2010,Minguzzi2013}, and~\cite{Minguzzi2026} for further representation results. The temporal and compactness properties needed here are proved directly below.

The paper is organized as follows. In Section~2 we recall the notion of a chronological space and present additional results on temporal foliations of causal sets, as well as the hypotheses and preliminary results used in the continuous case. In Section~3 we recall the main definitions in space-time causality and compare the notion of a timelike curve with the corresponding notion of a maximal chain. In Section~4 we proceed to the proof of the main result. In Section~\ref{sec:spacetimes} we obtain Geroch's Theorem (for a globally hyperbolic space-time) as a corollary of the main result; in particular, we show the existence of a continuous temporal level function, and also that the slices have spacelike or null tangent hyperplanes defined almost everywhere. In Section~\ref{sec:generalized} we show the existence of a generalized Cauchy hypersurface for chronological space-times.

This version corrects the arguments in~\cite{Bleybel2022}. Proposition 2.3 is retained with a new proof. The former Proposition 2.4 is omitted, the hypotheses in Section~2.3 are amended explicitly, and the discrete refinement argument is replaced by the construction in Section~4.

\section{Chronological spaces}
In this section I review some background results and notation on ordered sets that will be used throughout the text.

\subsection{Notation}
A chronological space is a tuple $(X,\ll)$, where $X$ is a set and $\ll$ is a strict partial order on $X$. A typical example of a chronological space (also called a ``chronological set'' elsewhere) is a strongly causal space-time $M$ with the chronology relation $\ll$ defined in the standard way.

A related notion is that of a causal space~\cite{KronheimerPenrose1967}. A causal space is a quadruple $(X,\prec,\ll,\to)$, where $X$ is a set and $\prec,\ll,\to$ are binary relations satisfying suitable compatibility axioms.

Two distinct elements $x$ and $y$ are incomparable if neither $x\ll y$ nor $y\ll x$; this is abbreviated using the notation $x\parallel y$.

We say that $X$ is full in case the following holds:
\begin{equation}
 \forall x,y,z\in X,\quad x,y\ll z\ \Longrightarrow\ \exists u\in X\ (x,y\ll u\ll z),
\end{equation}
and
\begin{equation}
 \forall x\in X,\quad \exists z\in X\ (z\ll x).
\end{equation}
If $A$ is a subset of $X$, then
\[
 A^\downarrow:=\{y\in X:(\exists x\in A)\ y\ll x\},\qquad
 A^\uparrow:=\{y\in X:(\exists x\in A)\ x\ll y\}.
\]
We also write $I^-(A)=A^\downarrow$ and $I^+(A)=A^\uparrow$. For $x\in X$, the set $I^-(x):=\{x\}^\downarrow$ is also denoted by $x^\downarrow$. Similarly, $I^+(x):=x^\uparrow:=\{x\}^\uparrow$.

The Alexandrov topology on $X$ is the smallest topology for which the sets $I^+(x)$ and $I^-(x)$ are open for all $x\in X$. Given two elements $x,y\in X$, we denote by $]x,y[$ (the open interval having endpoints $x$ and $y$, respectively) the set
\[
 ]x,y[:=\{z\in X:x\ll z\ll y\}=x^\uparrow\cap y^\downarrow.
\]
An antichain is a subset of $X$ whose distinct elements are mutually incomparable. A chain is a totally ordered subset of $X$. A maximal chain is a chain maximal with respect to inclusion. (So $C$ is maximal if for any other chain $C'$, $C\subseteq C'$ implies $C=C'$.) An antichain is also called achronal. A subset is relatively compact if its closure is compact.

Given two disjoint antichains $\Sigma_1,\Sigma_2$ in $X$, we say that $\Sigma_1\triangleleft\Sigma_2$ if any element of $\Sigma_1$ precedes some element of $\Sigma_2$, or else is incomparable to all elements of $\Sigma_2$, and furthermore, $\Sigma_1\cup\Sigma_2$ is not an antichain. The relation $\triangleleft$ is not necessarily transitive.

A foliation of a chronological space $X$ is a partition $F$ of $X$ into nonempty antichains (also called slices) $\Sigma_i$, $i\in I$, such that the relation $\triangleleft$ is a strict linear (or total) order relation on $F$. We retain this stronger definition from the original paper. In the terminology of the revised~\cite{BleybelZaiour2018}, it is a \emph{linked foliation}: every earlier slice has a direct chronological link to every later slice. The theorem below gives this stronger conclusion, so no change to this definition is needed here.

A foliation on $X$ is a temporal foliation if every antichain $\Sigma$ in the foliation satisfies the following condition:
\[
\begin{gathered}
 \text{for any maximal chain }C\subseteq X\text{ disjoint from }\Sigma,\\
 \text{there exist }x,y\in C\text{ and }z,t\in\Sigma
 \text{ such that }x\ll z\text{ and }t\ll y.
\end{gathered}\qquad(\ddagger)
\]
A slice satisfying this condition is called a quasi-Cauchy slice. If $F$ is a foliation of $X$, a slice of $F$ is more succinctly termed an $F$-slice. If $\Sigma$ is a top slice ($\Sigma^\uparrow=\varnothing$), then $\Sigma$ satisfies $(\ddagger)$ if and only if it intersects every maximal chain; the same applies to a bottom slice ($\Sigma^\downarrow=\varnothing$). This definition is order-theoretic; a local product structure is a further conclusion in the space-time case.

A causal set $X$ is a locally finite chronological space, i.e., the interval $]x,y[$ is finite for all $x,y\in X$. A causal set equipped with an irreflexive partial order $\ll$ can equivalently be described using the reflexive partial order $\leq:=\ll\,\cup\,\Delta$, where $\Delta$ is the diagonal of $X$.

\subsection{Some results on temporal foliations of causal sets}
In this subsection, we consider a causal set $\mathcal C$ equipped with a reflexive partial order $\prec$. Thus $x\prec x$ is allowed, and the closed interval $[x,y]=\{z:x\prec z\prec y\}$ is finite whenever $x\prec y$.

Our goal is to show the following.

\setcounter{theorem}{2}
\begin{proposition}\label{prop:discrete-extension}
Let $\mathcal C$ be a nonempty countable causal set, equipped with its
reflexive order $\prec$. Let $\Sigma_0\ne\varnothing$ be a finite antichain
in $\mathcal C$. Then there exists an antichain $\Sigma\supseteq\Sigma_0$
satisfying $(\ddagger)$ in its reflexive form: for every maximal chain
$C\subseteq\mathcal C$, there exist $x,y\in C$ and $z,t\in\Sigma$
with $x\prec z$ and $t\prec y$.
\end{proposition}

\begin{proof}
For $x\prec y$, let $\ell(x,y)$ denote the largest number of strict steps
in a chain from $x$ to $y$, and put $\ell(x,x)=0$. This number is finite
by local finiteness. If $x\prec y\prec z$, then
\begin{equation}\label{eq:length-superadditive}
 \ell(x,z)\geq\ell(x,y)+\ell(y,z).
\end{equation}

\medskip\noindent\textit{(i) Construction of the level function.}
We first construct a function $h:\mathcal C\longrightarrow\mathbb Z$
such that $h=0$ on $\Sigma_0$ and
\begin{equation}\label{eq:height-constraint}
 h(y)-h(x)\geq\ell(x,y)\qquad(x\prec y).
\end{equation}
Start with $h=0$ on $\Sigma_0$ and enumerate the remaining elements of
$\mathcal C$. Suppose that $h$ has been defined on a finite set $F$
containing $\Sigma_0$, subject to \eqref{eq:height-constraint}, and let
$u$ be the next element. The requirements on $h(u)$ are
\[
 h(u)\geq h(a)+\ell(a,u)\quad(a\in F, a\prec u),
 \qquad
 h(u)\leq h(b)-\ell(u,b)\quad(b\in F, u\prec b).
\]
Whenever both types of bound occur, they are compatible: for each such
$a,b$, \eqref{eq:length-superadditive} and the inductive assumption give
\[
 h(a)+\ell(a,u)\leq h(b)-\ell(u,b).
\]
There are only finitely many integer bounds. Thus an integer $h(u)$ can
be chosen; a missing lower or upper bound imposes no restriction on that
side. The induction defines the required function. In particular,
$x\prec y$, $x\ne y$, implies $h(x)<h(y)$.

\medskip\noindent\textit{(ii) Construction of the antichain.}
Put
\[
 H=\{x\in\mathcal C:h(x)\leq0\},\qquad M=\operatorname{Max}H.
\]
Every element of $H$ precedes an element of $M$. Indeed, for $x\in H$,
the nonempty set of integer values $h(y)$ with $x\prec y\in H$ has a
largest element; a point realizing that value is maximal in $H$.
Furthermore, $\Sigma_0\subseteq M$ by strict increase of $h$.

Let $R$ be the set of elements incomparable with every element of $M$,
and put $N=\operatorname{Min}R$; if $R$ is empty, set $N=\varnothing$.
Since every point of $H$ precedes $M$, one has $R\subseteq\{h>0\}$.
It follows that every point $x\in R$ succeeds a point of $N$: choose a
point of least height in the nonempty set $\{y\in R:y\prec x\}$.
Consequently
\[
 \Sigma=M\cup N
\]
is an antichain containing $\Sigma_0$. It is maximal, since points
outside $R$ are comparable with $M$, and points in $R$ succeed $N$.

\medskip\noindent\textit{(iii) Verification of $(\ddagger)$.}
Let $C$ be any maximal chain. If $C\cap H\ne\varnothing$, a point of
$C\cap H$ precedes an element of $M$, giving $x\prec z$ as required.
If $C\cap H=\varnothing$, then $h(C)$ is a nonempty set of positive
integers and has a least value. Since $h$ is strictly increasing on
$C$, the point attaining this value is the least element $c$ of $C$.
Maximality of $C$ implies that $c$ is minimal in $\mathcal C$. Since
$h(c)>0$, it is neither in $M$ nor below $M$, and minimality prevents
it from strictly succeeding $M$. Hence $c\in N\subseteq\Sigma$, and
the same requirement holds with $x=z=c$.

For the other requirement, suppose first that $C$ contains a point $y$
with $h(y)>0$. Such a point cannot precede $M$. It therefore either
succeeds an element of $M$, or belongs to $R$ and succeeds an element
of $N$. In both cases $t\prec y$ for some $t\in\Sigma$.
Finally, if $C\subseteq H$, then $h(C)$ has a largest value. The point
attaining it is the greatest element of $C$, and maximality of the
chain makes this point maximal in $\mathcal C$. It belongs to $M$,
and we may take $t=y$ equal to this point. Thus both requirements hold
for every maximal chain.
\end{proof}

The length $\ell(x,y)$ is computed in the whole causal set, including points to which a level has not yet been assigned. Thus each finite step leaves room for every intermediate point. The last part of the proof applies directly to an arbitrary maximal chain.

The following figures illustrate the role of local finiteness and of condition
$(\ddagger)$.  A finite antichain may lie above an infinite part of a chain if
local finiteness is dropped (Figure~\ref{fig:finite-antichain}).  An infinite
antichain need not satisfy $(\ddagger)$ even in a locally finite causal set
(Figure~\ref{fig:infinite-antichain}).  Figure~\ref{fig:quasi-cauchy-bracketing}
shows the required relations for a chain disjoint from the antichain.

\begin{figure}[htbp]
\centering
\includegraphics[width=.32\textwidth]{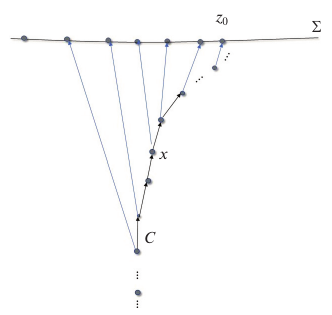}
\caption{A finite antichain $\Sigma$ and an infinite chain $C$ whose elements
lie to the past of $\Sigma$. The causal space $\Sigma\cup C$, with the order
relation given by the transitive closure of the arrows, is not locally finite.}
\label{fig:finite-antichain}
\end{figure}

\begin{figure}[htbp]
\centering
\includegraphics[width=.34\textwidth]{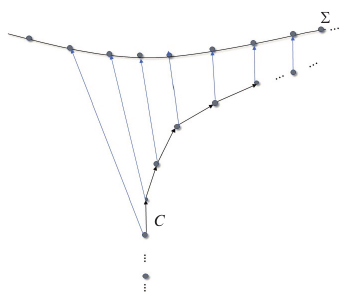}
\caption{An antichain $\Sigma$ and a maximal chain $C$ whose elements lie to
the past of $\Sigma$. Both $\Sigma$ and $C$ are infinite. The causal space
$\Sigma\cup C$ is locally finite, but condition $(\ddagger)$ is not met.}
\label{fig:infinite-antichain}
\end{figure}

\begin{figure}[htbp]
\centering
\includegraphics[width=.30\textwidth]{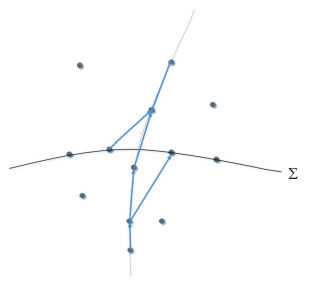}
\caption{The two chronological relations required by condition $(\ddagger)$ for a chain disjoint from the antichain.}
\label{fig:quasi-cauchy-bracketing}
\end{figure}

\subsection{Dense chronological spaces}
We now impose the following topological and order-theoretic conditions on $(X,\ll)$.

Let $\tau$ denote the chronological topology defined above. The canonical causal relation associated with the chronology is
\begin{equation}\label{eq:canonical-order}
x\leqC y\quad\Longleftrightarrow\quad
\Im(x)\subseteq\Im(y)\ \text{and}\ \Ip(y)\subseteq\Ip(x).
\end{equation}
Once this relation is a partial order, set
\[
\Jp(A)=\{y:(\exists a\in A)\ a\leqC y\},\qquad
\Jm(A)=\{y:(\exists a\in A)\ y\leqC a\}.
\]
All closures and all uses of compactness below refer to $\tau$.

\setcounter{theorem}{4}
\begin{assumption}[Ambient hypotheses]\label{ass:H}
The set $X$ is nonempty and satisfies:
\begin{enumerate}[label=(H\arabic*)]
\item $\tau$ is Hausdorff and second countable.
\item For each $x\in X$,
\[
x\in\cl{\Im(x)}\cap\cl{\Ip(x)}.
\]
\item The graph of $\leqC$ is closed in $X\times X$.
\item For all $p,q\in X$, the interval
\[
[p,q]_c=\Jp(p)\cap\Jm(q)
\]
is compact.
\end{enumerate}
\end{assumption}

The assumptions imply chronological interpolation and two-sided local accessibility. They will be used in Section~4. In particular, closedness of the canonical causal order and compactness of its intervals are explicit hypotheses.

Throughout the rest of this subsection, Assumption~\ref{ass:H} is in force.

\begin{lemma}[Order, interpolation and openness]\label{lem:order}
The relation $\leqC$ is a partial order containing $\ll$. It satisfies both push-up laws:
\[
x\leqC y\ll z\Longrightarrow x\ll z,
\qquad x\ll y\leqC z\Longrightarrow x\ll z.
\]
Whenever $x\ll y$, there exists $z$ with $x\ll z\ll y$. The graph of $\ll$ is open in $X\times X$. The closure of an achronal set is achronal.
\end{lemma}
\begin{proof}
Reflexivity and transitivity of $\leqC$ follow from set inclusion. If $x\leqC y\leqC x$, the chronological pasts and futures of $x$ and $y$ coincide. They then belong to exactly the same subbasic chronological open sets, hence to the same open sets. Hausdorffness implies $x=y$. Transitivity of $\ll$ gives $\ll\,\subseteq\,\leqC$.

For the first push-up implication, $z\in\Ip(y)\subseteq\Ip(x)$. For the second, $x\in\Im(y)\subseteq\Im(z)$. If $x\ll y$, the open neighborhood $\Im(y)$ of $x$ meets $\Ip(x)$ by (H2), giving interpolation. Repeating it gives $x\ll u\ll v\ll y$. Then $\Im(u)\times\Ip(v)$ is an open neighborhood of $(x,y)$ contained in the graph of $\ll$.

If $x,y$ belong to the closure of an achronal set $A$ and $x\ll y$, this product openness gives neighborhoods $U$ of $x$ and $V$ of $y$ such that every point of $U$ precedes every point of $V$. Points of $A\cap U$ and $A\cap V$ contradict achronicity. Hence $\cl A$ is achronal.
\end{proof}

\begin{lemma}[Compact sets and causal diamonds]\label{lem:compact-diamonds}
For compact $K,L\subset X$, the sets $\Jp(K)$ and $\Jm(K)$ are closed, and
\[
\Jp(K)\cap\Jm(L)
\]
is compact. Moreover,
\begin{equation}\label{eq:closure-chronological-compact}
\cl{\Ip(K)}=\Jp(K),\qquad
\cl{\Im(K)}=\Jm(K).
\end{equation}
Every compact set is covered by finitely many chronological futures of points and, dually, by finitely many chronological pasts of points.
\end{lemma}
\begin{proof}
Projection onto the second factor maps a closed subset of $K\times X$ to a closed subset of $X$ when $K$ is compact. Applied to the closed order graph, this proves closedness of $\Jp(K)$; the other case is dual. For completeness, if $y$ is outside such a projection, every $(k,y)$ has a product neighborhood avoiding the closed set. Finitely many of the neighborhoods of $k$ cover $K$, and the intersection of the corresponding neighborhoods of $y$ avoids the projection.

By (H2), every $k\in K$ has a predecessor $p_k\ll k$. Openness and compactness give finitely many $p_i$ with $K\subset\bigcup_i\Ip(p_i)$. Dually choose finitely many $q_j$ with $L\subset\bigcup_j\Im(q_j)$. If $k\leqC x\leqC l$ for $k\in K$ and $l\in L$, suitable anchors give $p_i\ll x\ll q_j$. Therefore
\[
\Jp(K)\cap\Jm(L)\subseteq\bigcup_{i,j}[p_i,q_j]_c.
\]
The left side is closed and the right side is compact by (H4), proving the claim.

Since $\Ip(K)\subseteq\Jp(K)$ and the latter is closed, one inclusion in the first equality of \eqref{eq:closure-chronological-compact} follows. Conversely, if $k\leqC x$, every point in $\Ip(x)$ lies in $\Ip(k)$ by push-up. As $x\in\cl{\Ip(x)}$, we have $x\in\cl{\Ip(K)}$. The other equality is dual.
\end{proof}

\begin{lemma}[Local compactness and countable tools]\label{lem:topology}
The space $X$ is locally compact, metrizable and normal. It has a countable dense subset and an increasing compact exhaustion $(K_n)_{n\ge1}$ such that
\[
K_n\subset\interior K_{n+1},\qquad \bigcup_n K_n=X.
\]
Every compact subset of $X$ is contained in some $K_n$.
\end{lemma}
\begin{proof}
Choose $p\ll x\ll q$ by (H2). The open neighborhood $\Ip(p)\cap\Im(q)$ of $x$ is contained in the compact closed interval $[p,q]_c$. Its closure is therefore compact. A locally compact Hausdorff space is regular; second countability then gives metrizability, and metrizable spaces are normal. A countable dense set is obtained by choosing a point from each nonempty member of a countable base.

Relatively compact open neighborhoods admit a countable subcover, since second-countable spaces are Lindelof. Starting with a finite collection of their compact closures, successively enlarge the collection so that its interiors cover the previous compact union and the first $n$ members of the countable cover. This gives a compact exhaustion with $K_n\subset\interior K_{n+1}$. The increasing interiors cover $X$, so a finite-subcover argument puts any compact subset inside one $K_n$.
\end{proof}

\begin{lemma}[Extrema in compact ordered sets]\label{lem:compact-extrema}
Every nonempty compact subset $E\subset X$, with its induced order, has a maximal and a minimal element.
\end{lemma}
\begin{proof}
For a chain $A\subset E$, the closed sets $E\cap\Jp(a)$, $a\in A$, have the finite intersection property: a largest element of any finite subset of $A$ belongs to the corresponding intersection. Compactness gives a common upper bound in $E$. Zorn's lemma therefore gives a maximal element of $E$. Apply the same argument to the reversed order for a minimal element.
\end{proof}

\begin{definition}[Past sets]
A past set is a set $P$ with $P=\Im(P)$. It is open, since chronological pasts are open. A subset is causally lower if $y\in P$ and $x\leqC y$ imply $x\in P$. Future sets and causally upper sets are defined dually.
\end{definition}

\begin{lemma}[Past sets and their boundaries]\label{lem:basic-past}
An open chronologically lower set is a past set. Every past set $P$ is causally lower and regular open, meaning $P=\interior\cl P$. Its closure is causally lower and its boundary is achronal. Furthermore,
\begin{equation}\label{eq:push-closure-past}
x\ll y\in\cl P\Longrightarrow x\in P.
\end{equation}
The set $F=X\setminus\cl P$ is an open future set and $\cl F=X\setminus P$. All these assertions have order-dual versions.
\end{lemma}
\begin{proof}
If $P$ is open and chronologically lower and $x\in P$, (H2) supplies $y\in P\cap\Ip(x)$; hence $x\in\Im(P)$. The reverse inclusion is the lower-set property. If $x\leqC y\in P$, choose $z\in P$ with $y\ll z$. Push-up gives $x\ll z$, so $x\in P$.

If $x\ll y\in\cl P$, the neighborhood $\Ip(x)$ of $y$ meets $P$. Thus $x\ll z\in P$ for some $z$, proving \eqref{eq:push-closure-past}. If $x\leqC y\in\cl P$, every $u\ll x$ satisfies $u\ll y$ and consequently $u\in P$. By (H2), $x\in\cl P$, so $\cl P$ is causally lower.

For $x\in\interior\cl P$, choose $y\in\Ip(x)\cap\interior\cl P$. Equation~\eqref{eq:push-closure-past} gives $x\in P$, proving regular openness. If $x,y\in\bd P$ and $x\ll y$, the same equation would give $x\in P$, impossible because $P$ is open. Hence $\bd P$ is achronal.

The complement $F$ of the closed lower set $\cl P$ is open and upper, so the dual of the first argument makes it a future set. The identity $\cl F=X\setminus\interior\cl P=X\setminus P$ follows from regular openness.
\end{proof}

\begin{lemma}[Compact bounds cannot contain maximal chains]\label{lem:finite-anchor}
Let $K\subset X$ be compact and let $C$ be a maximal chronological chain. Then $C$ has no first or last element and
\[
C\not\subset\Im(K),\qquad C\not\subset\Ip(K).
\]
\end{lemma}
\begin{proof}
A last point of $C$ could be followed by a local chronological successor, extending $C$; a first point is excluded dually. Cover $K$ by finitely many sets $\Im(q_i)$. If $C\subset\Im(K)$, then the sets $C\cap\Im(q_i)$ cover $C$. Each is an initial segment, and initial segments of a linear order are nested. One of these finitely many segments must equal $C$. Its $q_i$ strictly succeeds the whole chain and could be adjoined, contradicting maximality. The future assertion is dual.
\end{proof}

\begin{remark}
Lemma~\ref{lem:finite-anchor} needs only compactness of $K$, open chronological pasts and futures, and local chronological predecessors and successors. It does not use (H3) or (H4). This weaker form will be used for the generalized-surface theorem.
\end{remark}

\section{Causality in Lorentzian manifolds}
\label{sec:causality}
Let $(M,g)$ be a smooth, time-oriented Lorentzian manifold of dimension
$n\geq2$, with metric signature $(-,+,\ldots,+)$. As usual, the manifold
is assumed to be Hausdorff and second countable.

The chronological past of an event $p$ in $M$, denoted by $I^-(p)$,
is the set of all events $q\in M$ such that $p$ can be reached from $q$
using a future-directed timelike curve. Similarly, the chronological
future of $p$, $I^+(p)$, is the set of events which can be reached from
$p$ using such a curve. The causal past and future, $J^-(p)$ and
$J^+(p)$, are defined analogously by replacing ``timelike'' by
``causal'', and include $p$ itself.

The Alexandrov interval is the open set
\[
 I(p,q)=I^+(p)\cap I^-(q).
\]
The space-time is chronological if $p\not\ll p$ for all $p$, and is
strongly causal if every event has arbitrarily small neighborhoods which
no causal curve meets in disconnected pieces. A strongly causal
space-time is globally hyperbolic if each closed causal diamond
\[
 [p,q]=J^+(p)\cap J^-(q)
\]
is compact. Strong causality implies that the Alexandrov topology
coincides with the manifold topology. Global hyperbolicity also implies
that the causal relation is closed \cite{Penrose1972,Minguzzi2019}.

We denote the relation $q\in I^-(p)$ by $q\ll p$. Similarly,
$q\prec p$ is a shorthand for $q\in J^-(p)$; this is the reflexive
causal relation also denoted by $q\leq p$ below.

Let $X$ be a subset of $M$. Then $X$ is naturally equipped with the
induced chronological relation $\ll_X=\ll|_{X\times X}$. Henceforth,
the relation $\ll_X$ will be denoted by $\ll$ whenever there is no risk
of confusion. A similar convention applies to the causal relation.

The closest notion that resembles an inextendible timelike curve in a
chronological space is that of a maximal chain with respect to $\ll$.
Viewing a chronological space-time as a chronological space, we observe,
however, that these two notions do not coincide.

In the following a timelike curve is understood to be continuous,
future-directed and piecewise smooth. When two smooth pieces join,
their tangent vectors have the same time orientation. The curve is
inextendible if it has no endpoint in $M$ in either time direction.

\begin{proposition}
\label{prop:gh-curve-chain}
\label{rem:gh-curve-chain}
In a globally hyperbolic space-time, the image of every inextendible
timelike curve is a maximal chronological chain.
\end{proposition}

\begin{proof}
Let $C$ be the image of an inextendible timelike curve. Its points
are chronologically ordered, so $C$ is a chronological chain.
Suppose that an event $q\notin C$ could be adjoined to it. Then $q$
would be chronologically comparable with every point of $C$, and the
disjoint relatively open sets
\[
 C\cap I^-(q),\qquad C\cap I^+(q)
\]
would cover $C$. Since $C$ is connected, one of these sets is all of
$C$. In the first case, choose $p\in C$. The future tail of the curve
starting at $p$ lies in the compact diamond
$J^+(p)\cap J^-(q)$. This contradicts the fact that a strongly causal
space-time admits no future-inextendible causal curve imprisoned in
a compact set \cite[Proposition~4.80]{Minguzzi2019}. The other case follows
by reversing the time orientation. Thus no event outside $C$ can be
adjoined, and $C$ is maximal.
\end{proof}

An inextendible timelike curve which is also a maximal chronological
chain is called a \emph{maximal timelike curve}. In a strongly causal
space-time this remains an additional requirement. This is the class
of curves considered in Section~\ref{sec:generalized}.

For example, in two-dimensional Minkowski space-time, the chronological
chain
\[
 C=\{(t,0):t\leq0\}\cup\{(t,1):t>1\}
\]
is maximal but disconnected. A point $(s,y)$ with $0<s\leq1$ comparable
with its two parts would have to satisfy $|y|<s$ and
$|1-y|\leq1-s$, contradicting the triangle inequality. At other times,
the point already in $C$ at that time prevents an additional point.
Thus even a maximal chronological chain need not be the image of a
continuous timelike curve; see also~\cite{Foldes2012}.

\section{Temporal splitting for chronological spaces}
\subsection{Theorem}
Our main result is the following.

\begin{theorem}[Prescribed-antichain temporal levels]\label{thm:main}
Under Assumption~\ref{ass:H}, let $\Sigma_0\subset X$ be an achronal set contained in a compact set. There exists a continuous surjection $T:X\to\R$ with
\begin{equation}\label{eq:main-monotonicity}
\Sigma_0\subseteq T^{-1}(0),\qquad
x\leqC y\Longrightarrow T(x)\le T(y),\qquad
x\ll y\Longrightarrow T(x)<T(y).
\end{equation}
For each $t\in\R$, the level $S_t=T^{-1}(t)$ is nonempty, closed and maximal achronal, and
\begin{equation}\label{eq:main-boundary-identities}
\Im(S_t)=\{T<t\},\qquad \Ip(S_t)=\{T>t\}.
\end{equation}
For every compact $L\subset X$, both
\begin{equation}\label{eq:main-caps}
\Jp(L)\cap\{T\le t\},\qquad
\Jm(L)\cap\{T\ge t\}
\end{equation}
are compact. On every maximal chronological chain $C$, $T$ is unbounded above and below. In particular, every $S_t$ brackets every such chain:
\[
C\cap\Im(S_t)\ne\varnothing,
\qquad C\cap\Ip(S_t)\ne\varnothing.
\]
Every maximal chronological chain that is connected in $\tau$ meets every $S_t$ exactly once. The family $(S_t)_{t\in\R}$ is a temporal foliation in the order-theoretic sense.
\end{theorem}

The connectedness required in the intersection assertion concerns the individual chain. No connectedness assumption on $X$ is needed. The two monotonicity properties of $T$ are stated separately because the prescribed antichain may contain distinct points related by $\leqC$.

\subsection{Proof of the theorem}
We proceed as follows. First we construct a past set whose boundary contains the prescribed antichain. We then insert further past sets, retaining compact bounds in both time directions. Finally we complete the ordered family and obtain the required function.

\subsubsection{Proper past sets}
\label{sec:proper-cuts}

Throughout this proof, Assumption~\ref{ass:H} is in force.
All causal futures and pasts are taken with respect to $\leqC$.
An open past set means an open set $P$ with $P=\Im(P)$.
By Lemma~\ref{lem:basic-past}, such a set and its closure are causal lower
sets.  Conversely, every open causal lower set is an open past set.

\begin{definition}[Proper past cut]
An open past set $P\subset X$ is a \emph{proper past cut} if, for every compact
$K\subset X$, both sets
\begin{equation}
 \Jp(K)\cap\cl P,
 \qquad
 \Jm(K)\cap(X\setminus P)
 \label{eq:proper-caps}
\end{equation}
are compact.
\end{definition}

Here and below $\cl P$ abbreviates $\overline P$.  The two compactness
requirements concern a cut viewed from a compact set of events; they do not
require its whole boundary to be compact.  They are interchanged by reversing
chronology and replacing $P$ with $X\setminus\cl P$.

\subsubsection{An initial slice}

The initial cut requires no previously chosen foliation.  Let
$D=\{d_1,d_2,\ldots\}$ be a countable dense subset of $X$, and put
\begin{equation}
 a(x)=\min\{n:x\ll d_n\},\qquad
 b(x)=\min\{n:d_n\ll x\},\qquad
 f(x)=a(x)-b(x).
 \label{eq:rank-functions}
\end{equation}
Both minima exist: $\Ip(x)$ and $\Im(x)$ are nonempty open sets.  Push-up
gives
\begin{equation}
 x\leqC y\quad\Longrightarrow\quad
 a(x)\le a(y),\quad b(x)\ge b(y),\quad f(x)\le f(y).
 \label{eq:rank-monotonicity}
\end{equation}
No continuity or semicontinuity of these integer-valued functions is asserted.

\begin{lemma}[Rank cut]
\label{lem:rank-cut}
With the notation above, the set
\begin{equation}
 H_0=\{x:a(x)\le b(x)\},\qquad Q=\Im(H_0)
 \label{eq:rank-cut}
\end{equation}
defines a proper past cut.
\end{lemma}

\begin{proof}
The set $Q$ is open and is a past set by chronological interpolation.
Monotonicity shows that $Q\subset H_0$.

Fix a compact set $K$.  Choose $N$ so that
\begin{equation}
 K\subset\bigcup_{i=1}^N\Ip(d_i),
 \label{eq:past-dense-anchors}
\end{equation}
using compactness and the fact that the sets $\Ip(d_i)$ cover $X$.
Write $D_N=\{d_1,\ldots,d_N\}$.  We claim that
\begin{equation}
 \Jp(K)\cap\cl Q\subset \Jp(K)\cap\Jm(D_N).
 \label{eq:rank-future-cap-bound}
\end{equation}
Indeed, if $x\in\Jp(K)\cap\cl Q$, choose $k\in K$ and $i\le N$ with
$d_i\ll k\leqC x$.  Then $x\in\Ip(d_i)$ by push-up.  Every point
$z\in Q\cap\Ip(d_i)$ satisfies
$a(z)\le b(z)\le i\le N$, and therefore lies in $\Im(D_N)$.
As $\Ip(d_i)$ is a neighborhood of $x$, it follows that
$x\in\cl{\Im(D_N)}=\Jm(D_N)$, by
Lemma~\ref{lem:compact-diamonds}.  The right-hand side of
\eqref{eq:rank-future-cap-bound} is compact by the same lemma, and its left-hand
side is closed.  This proves the first compactness requirement.

For the other requirement, choose $N$ instead so that
\begin{equation}
 K\subset\bigcup_{i=1}^N\Im(d_i).
 \label{eq:future-dense-anchors}
\end{equation}
We claim that
\begin{equation}
 \Jm(K)\cap(X\setminus Q)
 \subset\Jm(K)\cap\Jp(D_N).
 \label{eq:rank-past-cap-bound}
\end{equation}
Let $x$ belong to the left-hand side.  For some $i\le N$, push-up gives
$x\ll d_i$.  Choose points $y_j\to x$ with
$x\ll y_j\ll d_i$; two-sided local chronological accessibility and
first countability allow this choice.  Since $x\notin\Im(H_0)$, none of the
$y_j$ belongs to $H_0$.  Thus
\[
 b(y_j)<a(y_j)\le i\le N,
\]
so $y_j\in\Ip(D_N)$.  Taking the limit gives
$x\in\cl{\Ip(D_N)}=\Jp(D_N)$.  Again, the right-hand side of
\eqref{eq:rank-past-cap-bound} is compact and the left-hand side is closed.
Both conditions in \eqref{eq:proper-caps} follow.
\end{proof}

\subsubsection{Retaining the prescribed antichain}

\begin{lemma}[Insertion of the prescribed seed]
\label{lem:seed-cut}
Let $Q$ be a proper past cut and let $K\subset X$ be compact and achronal.
Then
\begin{equation}
 P=(Q\setminus\Jp(K))\cup\Im(K)
 \label{eq:seed-cut}
\end{equation}
is a proper past cut and $K\subset\bd P$.
\end{lemma}

\begin{proof}
The set $Q\setminus\Jp(K)$ is open and causally lower, because $\Jp(K)$
is closed and causally upper.  It is therefore an open past set.
Consequently $P$ is an open past set.  No point of $K$ belongs to its first
term, and achronicity excludes $K\cap\Im(K)$.  On the other hand,
$K\subset\cl{\Im(K)}\subset\cl P$.  Thus $K\subset\bd P$.

For any compact $L$, the inclusions
\[
 \cl P\subset\cl Q\cup\Jm(K),
 \qquad
 X\setminus P\subset(X\setminus Q)\cup\Jp(K)
\]
give
\begin{align*}
 \Jp(L)\cap\cl P
 &\subset
 [\Jp(L)\cap\cl Q]\cup[\Jp(L)\cap\Jm(K)],\\
 \Jm(L)\cap(X\setminus P)
 &\subset
 [\Jm(L)\cap(X\setminus Q)]\cup[\Jm(L)\cap\Jp(K)].
\end{align*}
Each set on the right is compact, either by properness of $Q$ or by compactness
of compact-set causal diamonds.  Each set on the left is closed and is therefore
compact.
\end{proof}

If the prescribed achronal set is merely contained in a compact set, its
closure is compact and remains achronal because chronology is open.
Lemma~\ref{lem:seed-cut} can therefore be applied to that closure.

\subsubsection{The boundary of a proper past set}

\begin{lemma}[Boundary identities]
\label{lem:boundary-identities}
If $P$ is a proper past cut and $B=\bd P$, then $B$ is achronal and
\begin{equation}
 \Im(B)=P,\qquad \Ip(B)=X\setminus\cl P.
 \label{eq:boundary-identities}
\end{equation}
In particular,
\begin{equation}
 X=\Im(B)\ \dot\cup\ B\ \dot\cup\ \Ip(B),
 \label{eq:boundary-partition}
\end{equation}
and $B$ is a nonempty closed maximal achronal set.
\end{lemma}

\begin{proof}
If $b,c\in B$ and $b\ll c$, the open set $\Ip(b)$ meets $P$ near $c$.
The past-set property would then imply $b\in P$, a contradiction.  Hence
$B$ is achronal.  The same neighborhood argument shows
$\Im(B)\subset P$.  If $b\in B$, $b\ll x$, and $x\in\cl P$, it again
implies $b\in P$.  Thus $\Ip(B)\subset X\setminus\cl P$.

For the reverse past inclusion, take $x\in P$ and choose $y\in P$ with
$x\ll y$.  The nonempty set
\[
 E=\Jp(y)\cap\cl P
\]
is compact.  By Lemma~\ref{lem:compact-extrema}, it has a maximal element
$m$ for $\leqC$.  If $m\in P$, there is $z\in P$ with $m\ll z$.
Then $z\in E$ and $m<_c z$, contradicting maximality.  Therefore
$m\in\cl P\setminus P=B$.  Since $x\ll y\leqC m$, push-up gives
$x\ll m$, proving $P\subset\Im(B)$.

For the reverse future inclusion, take $x\in X\setminus\cl P$ and choose
$y\in X\setminus\cl P$ with $y\ll x$.  This is possible because the
complement of $\cl P$ is open and $x$ is approached from its chronological
past.  The set
\[
 E=\Jm(y)\cap(X\setminus P)
\]
is nonempty and compact, and has a minimal element $m$.
If $m\notin\cl P$, local past accessibility supplies
$z\in X\setminus\cl P$ with $z\ll m$.  This contradicts minimality
in $E$.  Consequently $m\in B$ and $m\leqC y\ll x$, so $m\ll x$.

The partition \eqref{eq:boundary-partition} follows because $P$ is open.
It implies maximal achronicity: every point outside $B$ is chronologically
related to a point of $B$.  It also implies that $B$ is nonempty, since $X$
is nonempty.
\end{proof}

\subsubsection{Enlarging a past set}

The next lemma supplies strict separation between successive cuts.  An
ordinary topological neighborhood of a past set need not have controlled
causal caps.  We first exclude a locally finite family of closed upper sets,
and only then take the chronological past of a neighborhood.

\begin{lemma}[Controlled thickening]
\label{lem:thickening}
Let $H$ be a closed causal lower subset of $X$ such that
\begin{equation}
 \Jp(K)\cap H\text{ is compact for every compact }K\subset X.
 \label{eq:H-future-caps}
\end{equation}
Suppose that $L$ is a proper past cut and $U$ is an open causal lower set,
with
\[
 L\subset H\subset U.
\]
Then there is a proper past cut $P$ satisfying
\begin{equation}
 H\subset P\subset\cl P\subset U.
 \label{eq:thickening-conclusion}
\end{equation}
\end{lemma}

\begin{proof}
Choose a sequence $a_1,a_2,\ldots$ with $\{\Ip(a_i)\}_{i\ge1}$ covering
$X$, and a compact exhaustion $K_i\subset\interior K_{i+1}$.
Such choices are available by Lemma~\ref{lem:topology} and local past
accessibility.  For each $i$ choose a compact set $E_i$ whose interior contains
both compact sets
\begin{equation}
 H\cap\Jp(a_i),\qquad \Jp(a_i)\cap\Jm(K_i).
 \label{eq:E-i-requirements}
\end{equation}
The first is compact by \eqref{eq:H-future-caps}, and the second by
Lemma~\ref{lem:compact-diamonds}.  Define
\begin{equation}
 D_i=\Jp(a_i)\setminus\interior E_i,
 \qquad F_i=\Jp(D_i).
 \label{eq:forbidden-upper-sets}
\end{equation}

We establish three properties of these forbidden sets.  First, $F_i$ is closed
and causally upper.  Upperness is immediate.  To see closedness, let
$x_j\in F_i$ converge to $x$, and choose $z_j\in D_i$ with
$z_j\leqC x_j$.  The convergent sequence together with its limit is a compact
set $A$, and all $z_j$ lie in the compact set
$\Jp(a_i)\cap\Jm(A)$.  Pass to a convergent subsequence, with limit $z$.
Since $D_i$ and the causal order are closed, $z\in D_i$ and $z\leqC x$.
Thus $x\in F_i$.  Sequences suffice because $X$ is metrizable by
Lemma~\ref{lem:topology}.

Second, $F_i\cap H=\varnothing$.  Indeed, if $z\in D_i$, $h\in H$, and
$z\leqC h$, the lower-set property of $H$ gives
$z\in H\cap\Jp(a_i)\subset\interior E_i$, a contradiction.
Third, $F_i\cap K_i=\varnothing$.  If $z\in D_i$, $k\in K_i$, and
$z\leqC k$, then
$z\in\Jp(a_i)\cap\Jm(K_i)\subset\interior E_i$, again a contradiction.

The family $\{F_i\}$ is locally finite.  For each point, choose a neighborhood
contained in some $K_N$.  All $F_i$ with $i\ge N$ miss that neighborhood,
because $K_N\subset K_i$.  Hence
\[
 F=(X\setminus U)\cup\bigcup_{i\ge1}F_i
\]
is a closed causal upper set disjoint from $H$.  By normality, there exists
an open set $O$ with
\begin{equation}
 H\subset O\subset\cl O\subset X\setminus F.
 \label{eq:thickening-neighborhood}
\end{equation}
Put $P=\Im(O)$.
It is an open past set, and $H\subset P$: if $h\in H\subset O$, local
future accessibility gives $o\in O$ with $h\ll o$.
Moreover $P\cap F=\varnothing$, since $F$ is upper and $O$ misses $F$.

We next prove the essential closure inclusion $\cl P\subset U$.
Fix $x\in\cl P$, and choose $i$ with $a_i\ll x$.
Choose $x_j\in P\cap\Ip(a_i)$ converging to $x$, and choose $o_j\in O$
with $x_j\ll o_j$.  Every $o_j$ lies in $\Jp(a_i)$, and no $o_j$ belongs
to $F_i$.  Since $D_i\subset F_i$, this forces $o_j\in\interior E_i$.
A subsequence converges to $o\in E_i\cap\cl O$.  Closedness of the order
gives $x\leqC o$.  By \eqref{eq:thickening-neighborhood}, $o\in U$;
as $U$ is lower, $x\in U$.

For properness, first observe that
\begin{equation}
 \cl P\cap\Ip(a_i)\subset E_i.
 \label{eq:thickening-local-bound}
\end{equation}
In fact $P\cap\Jp(a_i)$ misses $D_i\subset F_i$ and therefore lies in
$\interior E_i$; taking limits inside the open set $\Ip(a_i)$ proves
\eqref{eq:thickening-local-bound}.  Given a compact $K$, finitely many
$\Ip(a_i)$ cover $K$.  By push-up, the same finite family covers $\Jp(K)$.
Thus $\Jp(K)\cap\cl P$ is contained in a finite union of the compact
$E_i$.  It is closed, hence compact.  Finally,
\[
 \Jm(K)\cap(X\setminus P)
 \subset\Jm(K)\cap(X\setminus L),
\]
since $L\subset H\subset P$.  The right-hand side is compact by properness
of $L$, and the left-hand side is closed.  This proves the second cap condition
and completes the proof.
\end{proof}

The same statement holds with all time directions reversed.  In particular,
the complement $X\setminus\cl P$ of the closure of a proper past cut is a
proper open future cut, because open past sets are regular open
(Lemma~\ref{lem:basic-past}).

\begin{remark}[Two elementary cap estimates]
\label{rem:cut-cap-estimates}
If $P$ is a proper past cut and $K$ is compact, then
$\cl P\cup\Jm(K)$ is a closed causal lower set satisfying
\eqref{eq:H-future-caps}.  Indeed, its intersection with $\Jp(A)$, for compact
$A$, is the union of the proper cap $\Jp(A)\cap\cl P$ and the compact diamond
$\Jp(A)\cap\Jm(K)$.  Therefore Lemma~\ref{lem:thickening}, with $U=X$,
can enlarge a proper cut to contain both its old closure and $K$.
Time reversal gives the corresponding shrinking operation.

Also, if $V$ is a proper past cut and $K$ is compact, then
$V\cup\Im(K)$ is a proper past cut: its future caps lie in a union of one
proper cap and one compact diamond, while its past-complement caps are closed
subsets of those of $V$.  Finally, any open past set $R$ lying between two
proper past cuts $P_-\subset R\subset P_+$ is itself proper.  Its future
caps are closed subsets of those of $P_+$, and its past-complement caps are
closed subsets of those of $P_-$.
\end{remark}

\subsubsection{Insertion of slices}

\begin{lemma}[Strict insertion]
\label{lem:insertion}
Let $P_-$ and $P_+$ be proper past cuts with
$\cl{P_-}\subset P_+$.  If $p\in P_+\setminus\cl{P_-}$, there exists a
proper past cut $P$ such that
\begin{equation}
 \cl{P_-}\subset P\subset\cl P\subset P_+,
 \qquad p\in\bd P.
 \label{eq:strict-insertion}
\end{equation}
\end{lemma}

\begin{proof}
Set $H=\cl{P_-}$ and $U=P_+\setminus\Jp(p)$.  The set $H$ is closed,
causally lower, and has compact future caps.  The set $U$ is open and causally
lower.  Moreover $H\subset U$: if $p\leqC h\in H$, lowerness of $H$ would
give $p\in H$, contrary to the choice of $p$.
Apply Lemma~\ref{lem:thickening}, using $L=P_-$, to obtain a proper past
cut $V$ with
\[
 \cl{P_-}\subset V\subset\cl V\subset P_+\setminus\Jp(p).
\]
Define $P=V\cup\Im(p)$.  This is proper by the preceding cap estimate.
Because $\cl{\Im(p)}=\Jm(p)$ and $P_+$ is causally lower,
\[
 \cl P=\cl V\cup\Jm(p)\subset P_+.
\]
Finally, $p\notin V$ and $p\notin\Im(p)$, whereas
$p\in\cl{\Im(p)}$.  Hence $p\in\bd P$, and all inclusions in
\eqref{eq:strict-insertion} hold.
\end{proof}

\subsubsection{A countable family of slices}
\label{sec:completion}

We now use the insertion lemma to construct a countable family of proper past sets. Here properness always refers to \eqref{eq:proper-caps}.

\begin{lemma}[Cofinal, densely placed cuts]\label{lem:countable-cuts}
Let $P_0$ be a proper open past set. There is a family of proper open past
sets $(P_r)_{r\in\Q}$ containing this particular set at index $0$, such that
\begin{align}
 r<s&\ \Longrightarrow\ \cl{P_r}\subset P_s,
 \label{eq:strict-cut-nesting}\\
 \bigcup_{r\in\Q}P_r&=X,
 &\bigcap_{r\in\Q}\cl{P_r}&=\varnothing,
 \label{eq:cut-cofinality}\\
 \cl{\bigcup_{r\in\Q}\bd P_r}&=X.
 \label{eq:dense-boundaries}
\end{align}
Every compact set is contained in some $P_s$ and is disjoint from some
$\cl{P_r}$.
\end{lemma}

\begin{proof}
By Lemma~\ref{lem:topology}, choose a compact exhaustion $(K_n)_{n\geq1}$
with $K_n\subset\interior K_{n+1}$ and $X=\bigcup_n K_n$.
We first construct cuts $(D_j)_{j\in\Z}$, with $D_0=P_0$.

Suppose $D_{n-1}$ has been constructed, where $n\geq1$. The set
\[
 H_n=\cl{D_{n-1}}\cup\Jm(K_n)
\]
is closed and causally lower. For a compact $L$, the set
$\Jp(L)\cap H_n$ is compact: its first part is a compact cap of
$D_{n-1}$, and its second part is a compact diamond between compact
sets, by Lemma~\ref{lem:compact-diamonds}. Apply
Lemma~\ref{lem:thickening} with lower cut $D_{n-1}$ and upper set $X$.
It gives a proper cut $D_n$ such that
\[
 \cl{D_{n-1}}\cup K_n\subset D_n.
\]

For the negative indices apply the order-dual thickening lemma. If
$D_{-(n-1)}$ has been constructed, the open future set
$X\setminus\cl{D_{-(n-1)}}$ is a proper future cut by
Lemma~\ref{lem:basic-past}. The closed upper set
\[
 H'_n=(X\setminus D_{-(n-1)})\cup\Jp(K_n)
\]
has compact past caps and contains this future cut. Its thickening is a
proper open future set $V_n$ with $H'_n\subset V_n$. Set
\[
 D_{-n}=X\setminus\cl{V_n}.
\]
Regular openness and the complement identities in
Lemma~\ref{lem:basic-past} show that $D_{-n}$ is a proper open past set and
\[
 \cl{D_{-n}}=X\setminus V_n
       \subset D_{-(n-1)}\setminus K_n.
\]
Thus the integer family has strict closure nesting in the sense of
\eqref{eq:strict-cut-nesting}, covers $X$ by its open members, and has
empty intersection of its closures. Every compact subset of $X$ is
contained in some $K_n$, so it lies above $\cl{D_{-n}}$ and below $D_n$.

The family $(\bd D_j)_{j\in\Z}$ is locally finite. Indeed, the interior
of $K_n$ meets none of the boundaries with index $j\leq-n$ or $j\geq n$:
the former closures miss $K_n$, and $K_n\subset D_n\subset D_j$ for the
latter. Every point belongs to the interior of some $K_n$.

Let $(U_m)_{m\geq1}$ be a countable basis of nonempty open sets. At each
finite stage retain all cuts already constructed. If $U_m$ meets an
existing boundary, make no change. Otherwise choose $p\in U_m$.
The current family consists of the integer family with finitely many
insertions. It is therefore discretely ordered, with no endpoints, and
its cofinality shows that there are consecutive cuts $A,B$ such that
\[
 p\in B\setminus\cl A.
\]
For completeness, the cuts containing $p$ form a nonempty final segment
of this discrete order, bounded below by a cut whose closure misses
$p$. It has a first member $B$; its predecessor is $A$. The point $p$
belongs to neither $A$ nor $\bd A$, since $U_m$ meets no existing
boundary, and hence lies outside $\cl A$.
Apply Lemma~\ref{lem:insertion} to obtain a proper cut $E$ with
\[
 \cl A\subset E\subset\cl E\subset B,
 \qquad p\in\bd E.
\]
Insert $E$ between $A$ and $B$. Closure nesting with every other cut
follows from the existing order. Each finite-stage family of boundaries
remains locally finite, since only finitely many sets have been added
to the original locally finite family. In particular, this construction
does not require local finiteness of the union over all stages.

Let $\mathcal P$ be the countable ordered family of all cuts so obtained.
Every basis element meets one of its boundaries, so their union is
dense. The order on $\mathcal P$ is dense as well. To prove this, take
$A<B$ in $\mathcal P$. The boundary $\bd A$ is nonempty: otherwise
Lemma~\ref{lem:boundary-identities} would give both
$A=\Im(\bd A)=\varnothing$ and
$X\setminus\cl A=\Ip(\bd A)=\varnothing$, contrary to $X\ne\varnothing$.
Choose $a\in\bd A$. Since $a\in B$, local future accessibility and
openness of $B$ give $b\in B$ with $a\ll b$. Such a $b$ cannot belong
to $\cl A$, since $\Im(\cl A)\subset A$ by
Lemma~\ref{lem:basic-past}. Consequently $B\setminus\cl A$ is a nonempty
open set. A boundary $\bd E$ in $\mathcal P$ meets this set. If $E\leq A$
then $\bd E\subset\cl A$, and if $B\leq E$ then $\bd E\cap B=\varnothing$.
Both are impossible, so $A<E<B$.

The integer subfamily ensures that $\mathcal P$ has no first or last
element. Hence it is order-isomorphic to $\Q$, with $P_0$ assigned to
$0$. One can obtain the isomorphism directly by enumerating the two
countable orders, starting with the pair $(P_0,0)$, and alternately
extending a finite order-preserving matching to the next unused element
of each enumeration. Density and the absence of endpoints supply an
unused element in every required interval. The union of these finite
matchings is the desired bijection. Relabeling $\mathcal P$ by this
isomorphism proves all the assertions. The cofinality and compact-set
bounds survive relabeling because every integer cut is retained.
\end{proof}

\subsection{Completion of the foliation}

\begin{theorem}[Completion of a countable family of proper cuts]
\label{thm:completion}
Let $(P_r)_{r\in\Q}$ satisfy the conclusions of
Lemma~\ref{lem:countable-cuts}. Define
\begin{equation}
 T(x)=\inf\{r\in\Q:x\in P_r\}.
 \label{eq:definition-T}
\end{equation}
Then $T:X\to\R$ is a continuous surjection, causally nondecreasing and
chronologically strictly increasing. For rational $r$,
\begin{equation}
 P_r=\{T<r\},\qquad \cl{P_r}=\{T\leq r\},
 \qquad \bd P_r=T^{-1}(r).
 \label{eq:exact-rational-levels}
\end{equation}
For every real $t$, the set $P_t:=\{T<t\}$ is a proper open past set,
and its boundary $S_t=T^{-1}(t)$ satisfies
\begin{equation}
 \Im(S_t)=\{T<t\},\qquad \Ip(S_t)=\{T>t\}.
 \label{eq:real-level-identities}
\end{equation}
Every $S_t$ is a nonempty closed maximal chronological antichain.
Every maximal chronological chain $C$ has $T(C)$ unbounded above and
below. If $C$ is also connected in the ambient topology, it meets each
$S_t$ exactly once.
\end{theorem}

\begin{proof}
\emph{Finiteness and continuity.}
For each $x$, the set in \eqref{eq:definition-T} is nonempty by
cofinality, is a final segment of $\Q$, and is bounded below since
$x\notin\cl{P_q}$ for some $q$. Thus $T(x)$ is a finite real number.
For every real $t$,
\begin{align}
 \{T<t\}&=\bigcup_{r<t}P_r,
 \label{eq:T-open-sublevels}\\
 \{T>t\}&=\bigcup_{r>t}(X\setminus\cl{P_r}).
 \label{eq:T-open-superlevels}
\end{align}
The first identity follows directly from the defining infimum.
For the second, if $T(x)>t$, choose $t<r<T(x)$. Were
$x\in\cl{P_r}$, closure nesting would put $x$ in every $P_s$ with
$s>r$, giving $T(x)\leq r$, a contradiction. Conversely, if
$x\notin\cl{P_r}$, it belongs to no $P_q$ with $q\leq r$, so
$T(x)\geq r$. This proves the second identity and continuity.

\emph{Monotonicity.}
Every $P_r$ is causally lower. Hence $x\leqC y$ implies
$\{r:y\in P_r\}\subset\{r:x\in P_r\}$, and therefore
$T(x)\leq T(y)$.
To prove strict chronological increase, take $x\ll y$.
Two-sided local accessibility and openness of chronology supply a
nonempty open chronological interval between $x$ and $y$.
Density of the boundaries gives $a\in\bd P_r$ with $x\ll a\ll y$.
Apply the same argument to the nonempty interval between $a$ and $y$
to find $b\in\bd P_s$ with $a\ll b\ll y$.
Necessarily $r<s$: equality would contradict achronicity of
$\bd P_r$, while $s<r$ would imply $b\in P_r$ and then $a\in P_r$,
contrary to $a\in\bd P_r$. By
Lemma~\ref{lem:boundary-identities}, $x\in P_r$ and
$y\notin\cl{P_s}$. Consequently
\[
 T(x)\leq r<s\leq T(y).
\]
The use of two ordered boundary points is what gives a strict inequality.

\emph{The retained rational levels.}
If $x\in\bd P_r$, then $x\notin P_r$ and
$x\in\cl{P_r}\subset P_s$ for every $s>r$, so $T(x)=r$.
Conversely suppose $T(x)=r$. If $x\in P_r$, the identity
$P_r=\Im(P_r)$ gives $y\in P_r$ with $x\ll y$. Then
$r=T(x)<T(y)\leq r$, a contradiction. If
$x\notin\cl{P_r}$, the complementary open future-set identity gives
$z\notin\cl{P_r}$ with $z\ll x$. Then
$r\leq T(z)<T(x)=r$, again a contradiction.
Thus $x\in\bd P_r$. The same arguments show that points of $P_r$
have value strictly below $r$, and points outside $\cl{P_r}$ have
value strictly above $r$. This proves \eqref{eq:exact-rational-levels}.

\emph{Surjectivity.}
Fix $t\in\R$, and choose rationals $r_0<t<q$.
The proper-cut identities imply $P_{r_0}\ne\varnothing$;
choose $x\in P_{r_0}$. Take rational numbers $r_n$ in $(r_0,q)$
converging to $t$. Since
$x\in P_{r_n}=\Im(\bd P_{r_n})$, choose
\[
 z_n\in\bd P_{r_n}\cap\Ip(x).
\]
Every $z_n$ lies in the fixed compact cap
$\Jp(x)\cap\cl{P_q}$. This cap is sequentially compact: the ambient
space is metrizable by Lemma~\ref{lem:topology}.
A subsequence converges to $z$, and continuity and the already proved
rational-level identity give
\[
 T(z)=\lim T(z_n)=\lim r_n=t.
\]
Thus every real level is attained. No assertion about nonempty
irrational levels has been used to obtain this compactness argument.

\emph{All real cuts and their compact caps.}
For real $t$, let $P_t=\{T<t\}$. It is open and chronologically lower.
If $x\in P_t$, local future accessibility and openness of $P_t$ give
$y\in P_t$ with $x\ll y$. Hence $P_t=\Im(P_t)$.
Continuity gives $\cl{P_t}\subset\{T\leq t\}$. If $T(x)=t$, every
point of $\Im(x)$ has value below $t$, and
$x\in\cl{\Im(x)}$ by local accessibility. Therefore
\begin{equation}
 \cl{P_t}=\{T\leq t\},\qquad
 \bd P_t=\{T=t\}.
 \label{eq:real-cut-closure}
\end{equation}
For compact $L$, choose rationals $r<t<s$. Then
\[
 \Jp(L)\cap\cl{P_t}\subset
       \Jp(L)\cap\cl{P_s},
 \qquad
 \Jm(L)\cap(X\setminus P_t)\subset
       \Jm(L)\cap(X\setminus P_r).
\]
The sets on the right are compact. The sets on the left are closed,
by closedness of $\Jp(L)$ and $\Jm(L)$ and
\eqref{eq:real-cut-closure}. They are therefore compact as well.
Thus $P_t$ is proper. Lemma~\ref{lem:boundary-identities} now proves
\eqref{eq:real-level-identities}. Strict chronological increase makes
$S_t$ achronal, and these identities show that every point outside
$S_t$ is comparable with a point of $S_t$. Consequently $S_t$ is a
maximal chronological antichain. It is closed by continuity and
nonempty by surjectivity.

\emph{All maximal chains.}
Suppose a maximal chronological chain $C$ had $T$ bounded above.
Choose a rational $q$ larger than this bound and choose $c_0\in C$.
Then $C\subset P_q$, and the tail
$C\cap\Jp(c_0)$ lies in the compact set
\[
 K=\Jp(c_0)\cap\cl{P_q}.
\]
For any $c\in C$, absence of a greatest element gives a point
$d\in C$ strictly later than both $c$ and $c_0$.
Thus $d\in K$ and $c\ll d$, proving $C\subset\Im(K)$.
This contradicts Lemma~\ref{lem:finite-anchor}.
If $T$ were bounded below on $C$, choose rational $r$ below the bound.
The past tail through $c_0$ would lie in the compact set
$\Jm(c_0)\cap(X\setminus P_r)$; the dual argument would give
$C\subset\Ip(K)$, again impossible.

Hence $T(C)$ is unbounded in both directions. If $C$ is connected,
its continuous image $T(C)$ is an interval, and the only real interval
unbounded in both directions is $\R$. Each level is therefore met.
Since distinct members of $C$ are chronologically comparable and $T$
is chronologically strictly increasing, each level is met at most once.
\end{proof}

\begin{proof}[Proof of Theorem~\ref{thm:main}]
Let $K=\cl{\Sigma_0}$. It is compact because $\Sigma_0$ is contained
in a compact subset of the Hausdorff space $X$. It is achronal because
chronology is open: a pair of chronologically related points in its
closure would have neighborhoods containing a chronologically related
pair in $\Sigma_0$. Lemmas~\ref{lem:rank-cut} and~\ref{lem:seed-cut}
give a proper cut $P_0$ with $K\subset\bd P_0$.
Apply Lemma~\ref{lem:countable-cuts} and Theorem~\ref{thm:completion}.
The exact rational-level identity ensures
$\Sigma_0\subset K\subset T^{-1}(0)$.

The levels $(S_t)_{t\in\R}$ partition $X$. If $r<s$, then
\eqref{eq:real-level-identities} gives
\[
 S_r\subset\Im(S_s),\qquad S_s\subset\Ip(S_r).
\]
In particular, every point of the earlier level precedes some point
of the later level; reverse chronological relations are excluded by
strict increase of $T$. Thus the level order agrees with the slice
relation $\triangleleft$ used in the definition of temporal
foliation: $S_r\triangleleft S_s$ if and only if $r<s$.
For any maximal chronological chain $C$ and any $t$, choose points
of $C$ with $T$ below and above $t$. By
\eqref{eq:real-level-identities}, these points belong respectively to
$\Im(S_t)$ and $\Ip(S_t)$. This proves the required bracketing
condition, including for chains disjoint from the level. The connected
chain assertion is exactly the final assertion of
Theorem~\ref{thm:completion}.
\end{proof}

\section{A proof of Geroch's splitting theorem}
\label{sec:spacetimes}

In this section we specialize to the case where $M$ is a globally
hyperbolic space-time. We apply the preceding theorem to obtain the
topological splitting established by Geroch~\cite{Geroch1970}.
Throughout this section, a Cauchy hypersurface is understood to meet every
inextendible timelike curve exactly once; see \cite[Section~2]{BernalSanchez2006}.
We have:

\begin{lemma}
\label{lem:spacetime-hypotheses}
A globally hyperbolic spacetime satisfies the hypotheses of
Theorem~\ref{thm:main}, and its canonical relation \(\leqC\) is the
usual causal relation \(\leq\).
\end{lemma}

\begin{proof}
Strong causality makes the Alexandrov topology equal to the manifold
topology, which is Hausdorff and second countable. Local timelike
segments give two-sided local chronological accessibility.
Global hyperbolicity gives compact causal diamonds and a closed causal
relation \cite[Section~4.5.4]{Minguzzi2019}.

The push-up property implies that \(p\leq q\) entails
\(\Im(p)\subseteq\Im(q)\) and \(\Ip(q)\subseteq\Ip(p)\), and
therefore \(p\leqC q\). Conversely, suppose \(p\leqC q\).
Choose \(p_n\ll p\) converging to \(p\). The inclusion of chronological
pasts gives \(p_n\ll q\). Closedness of the usual causal relation
implies \(p\leq q\). Thus both orders and their compact intervals agree.
\end{proof}

\begin{theorem}
\label{thm:spacetime}
Let \((M,g)\) be a globally hyperbolic space-time. Let \(\Sigma_0\) be
an achronal set contained in a compact subset of \(M\). Then there
exists a foliation \(F\) of \(M\) by closed maximal achronal Cauchy
topological hypersurfaces such that \(\Sigma_0\) is contained in some
\(F\)-slice.
\end{theorem}

\begin{proof}
Apply Theorem~\ref{thm:main}, using
Lemma~\ref{lem:spacetime-hypotheses}. This gives a continuous surjection
$T:M\to\R$, with $\Sigma_0\subseteq T^{-1}(0)$, and closed maximal
achronal levels $S_t=T^{-1}(t)$. The function is strictly increasing
along the chronological order and nondecreasing along the causal order.
In addition, the sets
\begin{equation}
 J^+(L)\cap\{T\leq t\},\qquad
 J^-(L)\cap\{T\geq t\}
 \quad(L\subset M\text{ compact},\ t\in\R)
 \label{eq:spacetime-level-caps}
\end{equation}
are compact.

Let \(\gamma\) be an inextendible future-directed timelike curve.
If \(T\) were bounded above on its future tail starting at \(p\), that
tail would be contained in one of the compact sets
\(J^+(p)\cap\{T\leq t\}\). A strongly causal spacetime admits no
future-inextendible causal curve imprisoned in a compact set. The
time-dual fact applies to its past tail. Thus \(T\circ\gamma\) is
unbounded in both directions. Continuity gives an intersection with
every level, and strict chronological increase makes it unique.
The non-imprisonment assertion follows from strong causality
\cite[Proposition~4.80]{Minguzzi2019}.

It remains to show that each slice is a topological manifold of
dimension $n-1$. Fix \(p\in S_t\) and take a small coordinate cylinder
\((u,z)\), centered at \(p\), whose vertical coordinate curves are
future-directed timelike. For sufficiently small \(\epsilon>0\), the
points \((-\epsilon,0)\) and \((\epsilon,0)\) lie respectively in
\(\Im(p)\) and \(\Ip(p)\). Continuity of \(T\), followed by shrinking
the spatial coordinate ball, gives
\[
 T(-\epsilon,z)<t<T(\epsilon,z)
\]
for every \(z\) in that ball. Along each vertical line \(T\) is
continuous and strictly increasing. Hence the portion of \(S_t\) in
this cylinder is a unique graph \(u=h(z)\). Bracketing its values by
nearby vertical points shows that \(h\) is continuous. Projection onto
the spatial coordinate ball is therefore a homeomorphism from this
portion of \(S_t\), with inverse $z\mapsto(h(z),z)$.
Thus each $S_t$ is a topological hypersurface. The family
$F=\{S_t:t\in\R\}$ has the required properties.
\end{proof}

\subsection{Existence of time functions}
In this section we consider the issue of the existence of global time
functions on a globally hyperbolic space-time. The two monotonicity
properties below are distinguished because the prescribed achronal set
may contain null-related points.

\setcounter{theorem}{4}
\begin{theorem}
\label{thm:spacetime-time}
Let \(M\) be a globally hyperbolic space-time. Let \(\Sigma_0\) be an
achronal set contained in a compact subset of \(M\). Then there exists
a continuous surjection \(T:M\to\R\) such that
\[
 \Sigma_0\subseteq T^{-1}(0),\qquad
 p\leq q\Longrightarrow T(p)\leq T(q),\qquad
 p\ll q\Longrightarrow T(p)<T(q).
\]
For each \(t\in\R\), the level \(S_t=T^{-1}(t)\) is a closed maximal
achronal Cauchy hypersurface.
\end{theorem}

\begin{proof}
Take the function $T$ supplied by Theorem~\ref{thm:main} and used in
the proof of Theorem~\ref{thm:spacetime}. The first theorem supplies
continuity, surjectivity, the two monotonicity properties, and
$\Sigma_0\subseteq T^{-1}(0)$. The preceding proof shows that each
level is a Cauchy hypersurface, as required.
\end{proof}

\begin{remark}[Causal strictness and prescribed null pairs]
The function in Theorem~\ref{thm:spacetime-time} need not be a time
function in the convention requiring strict increase along every
nonconstant future-directed causal curve. For example, in two-dimensional
Minkowski spacetime, the compact achronal seed
\(\Sigma_0=\{(0,0),(1,1)\}\), in coordinates \((t,x)\), contains
distinct null-related points. Any function constant on this seed
fails causal strictness. Causal monotonicity in fact makes it constant
on the joining null segment. The function is therefore required to be
chronologically strictly increasing and causally nondecreasing. The
distinction for achronal prescribed sets also appears in
\cite{BernalSanchez2006}.
\end{remark}

\subsection{The slices are almost everywhere spacelike or null}
In this subsection we will show that the temporal splitting obtained
above consists of slices possessing spacelike or null tangent
hyperplanes defined almost everywhere. To achieve this, we make use of
the following form of Rademacher's theorem; see also~\cite{Lytchak2005}.

\begin{proposition}[Rademacher's theorem]
\label{prop:rademacher}
Let $U\subseteq\R^{n-1}$ be open, and let $h:U\to\R$ be locally
Lipschitz. Then $h$ is differentiable at almost every point of $U$
with respect to Lebesgue measure, and its differential at each such
point is a linear map from $\R^{n-1}$ to $\R$.
\end{proposition}

We obtain:

\subsection{Theorem}
\begin{theorem}
\label{thm:spacetime-regularity}
The foliation in Theorem~\ref{thm:spacetime} can be chosen so that each
slice is locally the graph of a Lipschitz function. Its tangent
hyperplane exists and is spacelike or null almost everywhere. Here
almost everywhere is understood with respect to Lebesgue measure in
the local graph coordinates.
\end{theorem}

\begin{proof}
Use the levels of the function $T$ in
Theorem~\ref{thm:spacetime-time}. Let $S_t$ be one such slice, and
consider its local continuous graph $u=h(z)$ constructed in the proof
of Theorem~\ref{thm:spacetime}. We will show first that $h$ is locally
Lipschitz.

After shrinking the cylinder once more, choose \(A>0\) such that every
coordinate vector \((v^0,v')\) with \(v^0>A|v'|\) is future timelike
throughout it. Keep the cylinder and its spatial base convex. If two
graph points satisfied
\(h(z_2)-h(z_1)>A|z_2-z_1|\), their straight coordinate segment would
be future timelike, contradicting achronicity. Interchanging the two
points gives
\[
 |h(z_2)-h(z_1)|\leq A|z_2-z_1|.
\]
Thus the graph is locally Lipschitz.
Rademacher's theorem gives differentiability almost everywhere. At a
differentiability point its tangent hyperplane contains no timelike
vector: a timelike tangent direction, by openness of the timelike
cones and differentiability of the graph, would give timelike-related
nearby graph points. In a Lorentzian vector space a hyperplane without
timelike vectors is spacelike or null. This proves the asserted
almost-everywhere geometric regularity.
\end{proof}

\medskip\noindent\textit{The product decomposition.}
There is also a homeomorphism
\[
 \Psi:\R\times S_0\longrightarrow M,\qquad
 T(\Psi(t,z))=t,
\]
such that, for each \(z\), the map \(t\mapsto\Psi(t,z)\) is an
orientation-preserving reparametrization of an inextendible smooth
timelike flow orbit. Differentiability with respect to the level
parameter \(t\) is not asserted.

To construct the product, fix a complete Riemannian metric \(h_0\) on
\(M\) and a smooth future-directed timelike vector field \(W\). The
rescaled field
\[
 V=\frac{W}{\sqrt{1+h_0(W,W)}}
\]
has bounded \(h_0\)-norm and therefore is complete. Write its flow as
\(\phi_u\). Its orbits are inextendible timelike curves: an endpoint
would contradict the local flow-box description of a nonvanishing
vector field. Each orbit consequently meets every \(S_t\) exactly once.
For \((t,z)\in\R\times S_0\), let \(\tau(t,z)\) be the unique number
such that
\[
 T(\phi_{\tau(t,z)}(z))=t,
 \qquad \Psi(t,z)=\phi_{\tau(t,z)}(z).
\]
The function \(\tau\) is continuous. In fact, if
\(u_0=\tau(t_0,z_0)\) and \(\delta>0\), strict increase gives
\[
 T(\phi_{u_0-\delta}(z_0))<t_0<
 T(\phi_{u_0+\delta}(z_0)).
\]
These inequalities persist for \((t,z)\) near \((t_0,z_0)\), and force
\(u_0-\delta<\tau(t,z)<u_0+\delta\). Hence \(\Psi\) is continuous.

For \(p\in M\), let \(\theta(p)\) be the unique flow time for which
\(T(\phi_{\theta(p)}(p))=0\). The same bracketing argument, with
\(p\) as the varying initial point, proves continuity of \(\theta\).
Thus
\[
 \Psi^{-1}(p)=\bigl(T(p),\phi_{\theta(p)}(p)\bigr)
\]
is continuous. This establishes the claimed homeomorphism, with every
\(\{t\}\times S_0\) mapped onto \(S_t\).

\section{Generalized Cauchy hypersurfaces}
\label{sec:generalized}

We now consider a generalized notion of a Cauchy hypersurface. Our aim
is to obtain a hypersurface crossed by every maximal timelike curve.
We assume only that the space-time is chronological; in particular, the
result applies to every strongly causal space-time. Closures and density
in this section refer to the manifold topology.

\begin{theorem}[Generalized Cauchy surface through a compact seed]
\label{thm:generalized}
Let \(M\) be a chronological spacetime and let \(\Sigma_0\subset M\) be a
compact achronal set. There is a closed maximal achronal, locally
Lipschitz hypersurface \(\Sigma\)
containing \(\Sigma_0\) such that
\begin{equation}
 M=\Im(\Sigma)\ \dot\cup\ \Sigma\ \dot\cup\ \Ip(\Sigma).
 \label{eq:generalized-partition}
\end{equation}
Equivalently, every $x\in M\setminus\Sigma$ is chronologically comparable
with a point of $\Sigma$. The exclusion of $x\in\Sigma$ is necessary
because $\Sigma$ is achronal.
Every connected maximal chronological chain meets \(\Sigma\) in exactly one
point. In particular, every inextendible timelike curve whose image is
a maximal chronological chain meets \(\Sigma\) exactly once. These assertions
hold, in particular, on every strongly causal spacetime.
\end{theorem}

We first prove two auxiliary lemmas.

\begin{lemma}[Compact anchors]
\label{lem:compact-anchors}
Let \(\Sigma_0\) be a compact subset of a chronological spacetime, and let \(C\)
be a maximal chronological chain. Then
\[
 C\not\subset \Im(\Sigma_0),\qquad C\not\subset \Ip(\Sigma_0),
 \qquad C\not\subset\cl{\Ip(\Sigma_0)}.
\]
\end{lemma}

\begin{proof}
The assertion is immediate if \(\Sigma_0=\varnothing\). Otherwise, local
chronological accessibility and compactness give events
\(q_1,\ldots,q_m\) such that
\(\Sigma_0\subset\bigcup_{i=1}^m\Im(q_i)\). If \(C\subset\Im(\Sigma_0)\), then
the initial segments \(C\cap\Im(q_i)\) cover \(C\). Any two initial
segments of a total order are nested, so one of these finitely many
segments is all of \(C\). The corresponding event \(q_i\) strictly
succeeds every point of \(C\). It does not belong to \(C\), by
irreflexivity, and adjoining it contradicts maximality. Reversing time
proves the second assertion.

We will also use the elementary implication
\begin{equation}
 x\in\cl{\Ip(\Sigma_0)},\quad x\ll y
 \quad\Longrightarrow\quad y\in\Ip(\Sigma_0).
 \label{eq:chronological-closure-push}
\end{equation}
Indeed, the open neighborhood \(\Im(y)\) of \(x\) meets \(\Ip(\Sigma_0)\),
and transitivity gives the conclusion. A maximal chronological chain
has neither a first nor a last point: a chronological predecessor of a
first point, or a chronological successor of a last point, could be
adjoined. If \(C\subset\cl{\Ip(\Sigma_0)}\), choose, for each \(y\in C\),
an \(x\in C\) with \(x\ll y\). Equation~\eqref{eq:chronological-closure-push}
would give \(C\subset\Ip(\Sigma_0)\), contrary to the second assertion.
\end{proof}

\begin{lemma}[A cut crossed by every maximal chain]
\label{lem:chronological-rank}
There is an open past set \(Q\) such that every maximal chronological
chain meets both \(Q\) and \(M\setminus Q\).
\end{lemma}

\begin{proof}
Choose a countable dense sequence \((d_n)_{n\geq1}\) and put
\[
 a(x)=\min\{n:x\ll d_n\},\qquad
 b(x)=\min\{n:d_n\ll x\},\qquad f(x)=a(x)-b(x).
\]
Both minima exist, since chronological futures and pasts are nonempty
open sets. Transitivity gives
\begin{equation}
 x\ll y\quad\Longrightarrow\quad
 a(x)\leq a(y),\quad b(x)\geq b(y),\quad f(x)\leq f(y).
 \label{eq:rank-chronological-monotonicity}
\end{equation}

Let \(C\) be any maximal chronological chain. The function \(a\) is
unbounded on \(C\): otherwise finitely many nested initial segments
\(C\cap\Im(d_n)\) would cover \(C\), and one \(d_n\) could be
adjoined above the whole chain. The dual argument shows that \(b\)
is unbounded on \(C\). Fix \(c_0\in C\). On the final segment above
\(c_0\), the function \(b\) is bounded by \(b(c_0)\), while \(a\)
is unbounded; on the initial segment below \(c_0\), the function \(a\)
is bounded by \(a(c_0)\), while \(b\) is unbounded. Consequently
\(f\) takes arbitrarily large positive and negative values on \(C\).

Set
\[
 H=\{x:f(x)\leq0\},\qquad Q=\Im(H).
\]
The set \(Q\) is an open past set, and
\eqref{eq:rank-chronological-monotonicity} gives \(Q\subset H\).
Choose \(h\in C\cap H\) and then \(c\in C\) with \(c\ll h\).
Thus \(c\in Q\). A point of \(C\) at which \(f>0\) lies outside
\(Q\). Thus $Q$ has the required property.
\end{proof}

\begin{proof}[Proof of Theorem~\ref{thm:generalized}]
Take \(Q\) from Lemma~\ref{lem:chronological-rank} and define
\begin{equation}
 P=\bigl(Q\setminus\cl{\Ip(\Sigma_0)}\bigr)\cup\Im(\Sigma_0),
 \qquad \Sigma=\bd P.
 \label{eq:generalized-seed-cut}
\end{equation}
By \eqref{eq:chronological-closure-push}, the complement of
\(\cl{\Ip(\Sigma_0)}\) is chronologically downward closed. It is open;
local future accessibility then shows that it is a past set.
Consequently \(P\) is an open past set.

Every \(k\in \Sigma_0\) belongs to both \(\cl{\Ip(\Sigma_0)}\) and
\(\cl{\Im(\Sigma_0)}\), by local accessibility. Achronicity gives
\(\Sigma_0\cap\Im(\Sigma_0)=\varnothing\). Thus \(\Sigma_0\cap P=\varnothing\) and
\(\Sigma_0\subset\cl P\), proving \(\Sigma_0\subset \Sigma\).

Fix a maximal chronological chain \(C\). The sets
\[
 C\cap Q,\qquad C\setminus\cl{\Ip(\Sigma_0)}
\]
are nonempty initial segments, by the two preceding lemmas and
\eqref{eq:chronological-closure-push}. They have nonempty intersection,
so \(C\cap P\ne\varnothing\). Likewise,
\[
 C\setminus Q,\qquad C\setminus\Im(\Sigma_0)
\]
are nonempty final segments. Their intersection is disjoint from \(P\).
We have thus shown that every maximal chain satisfies
\begin{equation}
 C\cap P\ne\varnothing,\qquad C\setminus P\ne\varnothing.
 \label{eq:generalized-chain-crossing}
\end{equation}

Let us show that the boundary of an open past set is achronal. If
\(s_1,s_2\in\bd P\) and \(s_1\ll s_2\), the open neighborhood
\(\Ip(s_1)\) of \(s_2\) meets \(P\). Downward closure then gives
\(s_1\in P\), a contradiction. If \(C\) is connected and avoided
\(\bd P\), its intersections with the disjoint open sets \(P\) and
\(M\setminus\cl P\) would separate it, contrary to
\eqref{eq:generalized-chain-crossing}. Thus \(C\) meets \(\Sigma\), and
achronicity makes the intersection unique.

It remains to establish the partition, and hence maximal achronicity.
For \(x\in P\), extend \(\{x\}\) to a maximal chronological chain
by Zorn's lemma. By \eqref{eq:generalized-chain-crossing}, this chain
contains \(y\notin P\), and necessarily \(x\ll y\). A timelike path
from \(x\) to \(y\) crosses \(\bd P\) after leaving \(x\); hence
\(x\in\Im(\Sigma)\). Conversely, if \(x\ll s\in \Sigma\), openness of
\(\Ip(x)\) and \(s\in\cl P\) imply \(x\in P\). Thus
\(\Im(\Sigma)=P\).

For \(x\notin\cl P\), a maximal chronological chain through \(x\)
has a point \(y\in P\), necessarily with \(y\ll x\). A timelike
path from \(y\) to \(x\) crosses \(\Sigma\) before reaching \(x\), so
\(x\in\Ip(\Sigma)\). Conversely, if \(s\ll x\) with \(s\in \Sigma\),
then \(x\notin\cl P\): otherwise the open neighborhood \(\Ip(s)\)
of \(x\) would meet \(P\), forcing \(s\in P\). Therefore
\(\Ip(\Sigma)=M\setminus\cl P\), proving
\eqref{eq:generalized-partition}. Every event outside \(\Sigma\) is
chronologically comparable with an event of \(\Sigma\), so \(\Sigma\) is
maximal achronal.

Finally, fix \(s\in \Sigma\) and take a local timelike flow-box cylinder
\((u,z)\) centered at \(s\). By the established identities
\(\Im(\Sigma)=P\) and \(\Ip(\Sigma)=M\setminus\cl P\), a small point on the
central vertical line below \(s\) lies in \(P\), and a small point
above it lies outside \(\cl P\). These two conditions are open, so
after shrinking the spatial base the entire lower and upper disks
have the respective properties. Every intervening vertical segment
crosses \(\Sigma\), exactly once by achronicity. The resulting graph is
continuous, since points strictly below and above a graph point lie
in the open sets \(P\) and \(M\setminus\cl P\). The uniform timelike
cone estimate in the proof of Theorem~\ref{thm:spacetime-regularity} applies to
this graph without using \(T\), and proves that it is locally
Lipschitz. Thus \(\Sigma\) is a locally Lipschitz hypersurface.
\end{proof}

\begin{remark}
Only the prescribed set \(\Sigma_0\) is assumed compact. The resulting surface
\(\Sigma\) may be noncompact. Compactness enters through the finite anchors
of Lemma~\ref{lem:compact-anchors}, which show that the modification
\eqref{eq:generalized-seed-cut} preserves the two-sided chain-crossing
property. The same conclusion
holds for an achronal seed contained in a compact set, by applying the
theorem to its compact achronal closure.
\end{remark}

\section*{Acknowledgments}
I wish to thank the anonymous referee for carefully reading the manuscript
and providing helpful comments. In particular, the generalization of
Geroch's Theorem to the setting of chronological spaces was motivated by
the referee's suggestions.
\section*{Conflict of interest}
The author declares that he has no conflicts of interest.
\ifdefined\withrevisionnotes
\clearpage
\appendix
\section{Corrections to the 2022 paper}
\label{sec:revision-notes}

This appendix distinguishes the statements proved here from the
arguments and hypotheses printed in \cite{Bleybel2022}.
Proposition~\ref{prop:discrete-extension} is retained with a complete
integer-valued level construction under its original hypotheses. The
continuous splitting theorem is proved by a direct construction in the
ambient chronological space. The former discrete approximation argument
and Proposition 2.4 are not used. The continuous proofs also do not invoke
the foliation theorem of \cite{BleybelZaiour2018}. Consequently, corrections
to that theorem or to its definitions do not enter the present proof.

\subsection{Changes in hypotheses and dependencies}

The compactness requirement in \cite[\S2.3, p.~365]{Bleybel2022} asserts
nonemptiness of the intersection of every nested sequence of closed sets.
Even after inserting the necessary qualification that those sets be nonempty,
this is not a property of globally hyperbolic spacetimes.  For example, in
Minkowski spacetime the sets
\[
 F_n=\{(t,\mathbf x):t\ge n\},\qquad n\in\N,
\]
are nonempty, closed and nested, but have empty intersection.  Compactness of
causal intervals does not imply compactness of these unbounded sets.

The abstract theorem therefore uses an amended hypothesis
package: a Hausdorff second-countable chronological topology, local
chronological accessibility in both directions, a closed canonical causal
order, and compact causal intervals.  These hypotheses are sufficient for the
proof given here.  No claim is made that the original printed assumptions
imply them, or that this package is minimal.

The prescribed antichain need only be contained in a compact set.  Its closure
is then compact and achronal, and the construction preserves that closure.
Compactness of the prescribed set must not be transferred without proof to a
completed, usually noncompact, slice.

\begin{center}
\begin{tabular}{@{}>{\raggedright\arraybackslash}p{0.30\linewidth}>{\raggedright\arraybackslash}p{0.65\linewidth}@{}}
\toprule
Original component & Treatment in this replacement\\
\midrule
Proposition 2.3 & Retained with a complete new proof under the original
countability and local-finiteness assumptions, preserving the prescribed
finite antichain by an integer-valued level construction.\\[4pt]
Proposition 2.4 & Omitted. Its pairwise compatibility hypothesis is
insufficient, as the finite counterexample below shows.\\[4pt]
Lemma 4.3 and causal-set refinements & Replaced by proper-cut bounds in the
ambient space.  This replacement does not assert that arbitrary prescribed
unbounded slices survive arbitrary refinements.\\[4pt]
Main chronological-space theorem & Reproved under the explicit amended
hypotheses, retaining the prescribed relatively compact antichain.\\[4pt]
Spacetime applications & Derived from the new abstract theorem and the
geometric arguments in the spacetime section.\\[4pt]
Theorem 6.1 & Replaced by the independent generalized-surface argument.
Its compact-seed control does not require compact causal intervals.\\
\bottomrule
\end{tabular}
\end{center}

The replacement restores Proposition 2.3 and proves the stated continuous
existence results. It does not restore every intermediate statement of the
original discrete refinement argument. In particular, Proposition 2.4 remains
omitted: retaining it would require a stronger compatibility hypothesis and
a separate proof.

\subsection{The integer-valued proof of Proposition 2.3}

The proof in Section~2.2 defines $\ell(x,y)$ using the longest chains
in the full ambient causal set, and constructs an integer-valued function
$h$ with $h=0$ on the prescribed finite antichain and
\[
 h(y)-h(x)\geq\ell(x,y)\qquad(x\prec y).
\]
The finite-stage bounds consequently account for intermediate points not
yet assigned a level. Superadditivity of $\ell$ makes the lower and
upper bounds compatible at every stage. Local finiteness ensures that
all the required ambient chain lengths are finite.

The final antichain is obtained from maximal elements of $\{h\leq0\}$
and minimal elements among the points incomparable with those maximal
elements. Its quasi-Cauchy property is then proved for an arbitrary maximal
chain of the entire causal set. Bounded nonempty sets of integer levels
have extrema, which handle chains lying on one side of the threshold,
including chains with endpoints. Thus the proof does not infer control of
all maximal chains from the list of chains treated at finite stages.
The reflexive order in Proposition 2.3 permits a witness to be an
intersection point of the chain and the slice.

\subsection{A finite obstruction to Proposition 2.4}

Let the six-element poset consist of
\[
 A=\{a_1,a_2\},\qquad B=\{b_1,b_2\},\qquad D=\{d_1,d_2\},
\]
with the only strict relations
\[
 a_1<b_1,\qquad b_2<d_1,\qquad d_2<a_2.
\]
These relations are already transitive: no endpoint of one displayed edge
is the initial point of another.  The three sets are disjoint antichains.
Between each pair of them, all existing relations have the same direction,
so the pairwise hypothesis of Proposition 2.4 is satisfied.

If a foliation retained each of $A,B,D$ inside a single slice, these three
slices would be distinct, since each pair contains comparable points.  The
three displayed relations would force their order to satisfy
\[
 S_A\mathrel{\triangleleft}S_B
 \mathrel{\triangleleft}S_D
 \mathrel{\triangleleft}S_A,
\]
contradicting the requirement that the slice order be a strict linear order.
Thus the first assertion of Proposition 2.4 is false as printed.  Pairwise
absence of oppositely directed relations does not exclude a cycle involving
three prescribed pieces.  Its additional quasi-Cauchy extension condition
does not remove this example: the antichains
\[
 A\cup\{b_2\},\qquad B\cup\{d_2\},\qquad D\cup\{a_1\}
\]
each meet every maximal chain of this poset.  Indeed, its maximal chains are
exactly the three two-element chains displayed above.  Thus both assertions
require a stronger compatibility hypothesis if they are to be retained.
The proposition is omitted from the present text; its number is reserved
to keep the remaining correspondence with the 2022 paper explicit.

\subsection{Bracketing a maximal chain versus intersecting it}

The definition of temporal foliation in \cite[\S2.1]{Bleybel2022} requires
two-sided chronological witnesses for a maximal chain disjoint from a slice.
This does not say that every maximal chronological chain intersects every
slice.  Connectedness is essential for the intersection conclusion used here.

In two-dimensional Minkowski spacetime, with coordinates $(t,x)$, consider
\[
 C=\{(t,0):t\le0\}\ \cup\ \{(t,1):t>1\}.
\]
This is a chronological chain.  It is maximal: a point $(s,y)$ with $s\le0$
or $s>1$ is incomparable with the point of $C$ at the same time unless it is
that point.  If $0<s\le1$, comparability with $(0,0)$ and with every $(t,1)$,
$t>1$, would require
\[
 |y|<s,\qquad |1-y|\le1-s.
\]
Their sum contradicts $|y|+|1-y|\ge1$.

Nevertheless $C$ misses the horizontal slice $t=1/2$.  The order isomorphism
\[
 (t,0)\longmapsto t\quad(t\le0),\qquad
 (t,1)\longmapsto t-1\quad(t>1)
\]
maps $C$ onto $\R$.  Hence even order density and Dedekind completeness of
the chain do not repair its failure of connectedness in the ambient
topology. This is an instance of the obstruction described before
\cite[Proposition~3.1]{Foldes2012}, which proves that the chronological
order of Minkowski space has no antichain cutset, that is, no antichain
meeting every maximal chronological chain.

The abstract conclusion is consequently stated in two parts: all maximal
chronological chains are bracketed by every level, and connected maximal
chronological chains intersect every level exactly once.  In the globally
hyperbolic spacetime application, the relevant inextendible timelike curves
are maximal chronological chains, so the usual timelike-curve crossing
conclusion follows.

\subsection{Corrections to the generalized-surface statement}

The first clause of \cite[Theorem 6.1]{Bleybel2022} requires a point of the
surface to be chronologically comparable with another point of that same
achronal surface.  It must instead be restricted to points outside the
surface.  The replacement proves the stronger explicit decomposition
\[
 M=I^-(S)\ \dot\cup\ S\ \dot\cup\ I^+(S).
\]

Its crossing assertion concerns inextendible timelike curves whose images
are also maximal chronological chains.  This is the intended generalized
Cauchy property. It is not replaced by an assertion about all inextendible
timelike curves in an arbitrary strongly causal space-time. Nor is every
maximal chronological chain assumed connected. For $x\in S$, achronicity
gives
\[
 (I^+(x)\cup I^-(x))\cap S=\varnothing.
\]

The new argument also avoids assuming that a nonempty open initial interval
of a timelike curve has a last point.  Neither replacing a maximum by a
supremum nor invoking compactness of the original seed alone controls a
noncompact completed surface.  Here the rank construction provides
simultaneous control of all maximal chains, and compactness is used only to
preserve that control when the prescribed seed is inserted. Closedness of
the causal relation, compactness of causal diamonds, and curvature bounds
on the resulting surface are not hypotheses of this argument.

\subsection{Terminology and relation to existing results}

An achronal seed may contain distinct null-related events.  If all of them
are retained in $T^{-1}(0)$, the function $T$ cannot be strictly increasing
on every nonconstant future-directed causal curve.  The conclusions are
therefore stated explicitly: causal nondecrease and strict chronological
increase.  The phrase ``temporal foliation'' is used in the order-theoretic
sense of the original paper; it does not assert existence of a smooth
temporal function with timelike gradient through an arbitrary achronal seed.

For ordinary globally hyperbolic spacetimes, compact achronal extension to
a Cauchy hypersurface is already established by
\cite[Proposition 3.6]{BernalSanchez2006}; that paper also provides a
construction confined between two Cauchy levels
\cite[Theorem 3.8]{BernalSanchez2006}.  Those are relevant antecedents to the
spacetime conclusion and the barrier idea.  The distinction between an
achronal prescribed level and causal strictness of the resulting function
is likewise explicit there. The spacetime consequence established here is
a topological product with locally Lipschitz levels. Smoothness of an
arbitrary prescribed level and an orthogonal metric splitting are not
asserted.

The causal-space framework goes back to \cite{KronheimerPenrose1967}.
Order-theoretic reconstruction from countably many events and domain
structures is developed in \cite{MartinPanangaden2006}.  General achronal
boundary and spacetime causality background can be found in
\cite{Penrose1972,Minguzzi2019}.  These references provide context; no
domain-completion theorem is needed as an unproved step in the construction
given here.

The order-separation and utility-representation results of
\cite{Minguzzi2010,Minguzzi2013} are relevant antecedents to the use of order
and topology. The present proofs establish their own past-set bounds,
insertion, and completion statements; they do not invoke a utility-existence
theorem. For additional representation results using time functions,
see \cite{Minguzzi2026}.

The purpose of this replacement is to give corrected statements and complete
proofs with the dependencies specified above. The absence of volume functions
describes the construction; it is not used as an argument for priority.

\fi
\phantomsection

\end{document}